\documentclass[aps,prx,reprint,twocolumn, longbibliography,superscriptaddress,floatfix,]{revtex4-1}

\usepackage{graphicx}
\usepackage{amsmath}
\usepackage{amssymb}
\usepackage{amsthm}
\usepackage{algpseudocode}
\usepackage{algorithm}

\usepackage[colorlinks = true]{hyperref}
\usepackage{xcolor}
\definecolor{darkred}  {rgb}{0.5,0,0}
\definecolor{darkblue} {rgb}{0,0,0.5}
\definecolor{darkgreen}{rgb}{0,0.5,0}

\hypersetup{
  urlcolor   = blue,         
  linkcolor  = darkblue,     
  citecolor  = darkgreen,    
  filecolor  = darkred       
}
\usepackage{cleveref}

\newcommand{\be}{\begin{equation}}
\newcommand{\ee}{\end{equation}}
\newcommand{\ba}{\begin{array}}
\newcommand{\ea}{\end{array}}
\newcommand{\bea}{\begin{eqnarray}}
\newcommand{\eea}{\end{eqnarray}}

\newcommand{\ra}{\rangle}
\newcommand{\la}{\langle}

\newcommand{\calL}{{\cal L }}

\newcommand{\calC}{{\cal C }}

\newcommand{\ZZ}{\mathbb{Z}}
\newcommand{\CC}{\mathbb{C}}

\newcommand{\RR}{\mathbb{R}}
\newcommand{\norm}[1]{\left\Vert#1\right\Vert}

\newtheorem{dfn}{Definition}

\newtheorem{lemma}{Lemma}
\newtheorem{fact}{Fact}
\newtheorem{claim}{Claim}

\newtheorem{corol}{Corollary}

\begin{document}

\title{Identity check problem for shallow quantum circuits}

\author{Sergey Bravyi}
\affiliation{IBM Quantum, IBM T.J. Watson Research Center, Yorktown Heights, NY 10598 (USA)}
\author{Natalie Parham}
\affiliation{Columbia University}
\author{Minh Tran}
\affiliation{IBM Quantum, IBM T.J. Watson Research Center, Yorktown Heights, NY 10598 (USA)}

\begin{abstract}
Checking whether two quantum circuits are approximately equivalent is  a common task in quantum computing.
We consider a closely related identity check problem: given a quantum circuit $U$, 
one has to estimate the diamond-norm distance between $U$ and the identity channel.
We present a classical algorithm approximating the distance to the identity 
within a factor $\alpha=D+1$ for shallow geometrically local $D$-dimensional circuits 
provided that the circuit is sufficiently close to the identity.
The runtime of the algorithm
scales linearly  with the number of qubits for any constant circuit depth and spatial dimension.
We also show that the operator-norm distance to the identity $\|U-I\|$ can be efficiently approximated within a
factor $\alpha=5$ for shallow 1D circuits and, under a certain technical condition,
within a factor  $\alpha=2D+3$ for shallow $D$-dimensional circuits. 
A numerical implementation of the  identity check  algorithm is reported for 1D Trotter circuits  with up to 100 qubits.
\end{abstract}

\maketitle

\section{Introduction}
\label{sec:intro}

A quantum circuit implementation of the desired unitary operation is rarely exact.
Common sources of errors include 
hardware noise owing to an imperfect control and decoherence,
errors introduced by the circuit compiling step,
and errors owing to an approximate nature of a quantum algorithm such as 
Trotter errors in simulation of Hamiltonian dynamics.
To validate a solution offered by a quantum
algorithm, is it essential that errors of each type are accounted for and
reasonably tight upper bounds on the
deviation from the ideal solution
are provided. 

Unfortunately, there is little hope that the distance between arbitrary  $n$-qubit quantum operations
can be computed efficiently for $n\gg 1$.
To begin with, an exponentially large Hilbert space dimension prevents one from 
obtaining the full matrix description of quantum operations or performing linear algebra
on such matrices. Furthermore, computational complexity theory provides no-go theorems for an
efficient distance estimation in many cases of interest.
For example, Rosgen and Watrous showed~\cite{rosgen2005hardness,rosgen2007distinguishing} that estimating the distance between two shallow (with depth logarithmic in $n$)
quantum circuits allowing mixed states 
is PSPACE-hard.
This essentially rules out efficient classical or quantum algorithms for the problem.
Likewise, Janzing, Wocjan, and Beth established QMA-hardness of estimating the distance between
two unitary circuits~\cite{janzing2005non}. The latter result was strengthened by Ji and Wu~\cite{ji2009non}
who proved QMA-hardness of estimating the distance between two constant-depth
circuits with the one-dimensional qubit connectivity. This may come as a surprise since
one-dimensional shallow circuits are easy to simulate classically using Matrix Product States~\cite{vidal2004efficient}.

It is important that the no-go results stated above hold only if the distance between quantum circuits
has to be estimated with a small {\em additive error} scaling inverse polynomially with the number of qubits $n$.
Is it possible that  some less stringent approximation of the distance can be computed efficiently?
Here, we show that the answer is YES
and report linear-time classical  algorithms approximating the 
diamond-norm and the
operator-norm distances 
between certain quantum circuits 
with a constant {\em multiplicative  error}.  Such approximation may be good enough for
practical purposes. Note that an estimate of the distance with a constant multiplicative error is informative
regardless of how small the distance is. For example, our algorithm can efficiently approximate the distance
even if the latter is exponentially small in $n$. This would be impossible for an algorithm that achieves
an additive error approximation scaling inverse polynomially with $n$.

Let us formally pose the distance estimation problem and
state our main results. Suppose $U$ is a unitary operator implemented by 
a quantum circuit acting on $n$ qubits. The diamond-norm distance~\cite{aharonov1998quantum} between $U$
and the identity operation  is defined as 
\be
\label{diamond_norm}
\delta(U) = \max_\rho \| (U\otimes I)\rho (U^\dag \otimes I) - \rho\|_1
\ee
where $\|\cdot \|_1$ is the trace  norm,
 $I$ is the $n$-qubit identity, and the maximization is over all $2n$-qubit states $\rho$.
The distance $\delta(U)$ has a simple operational meaning: replacing $U$ by the identity in any experiment
that makes use of one copy of $U$ 
 could change the probability distribution describing classical outcomes of the experiment at most by
$\delta(U)/2$ in the total variation distance~\cite{aharonov1998quantum,ben2009complexity}.
Accordingly $\delta(U)\le 2$  with the equality if $U$ is perfectly distinguishable from the identity
in the single-shot setting.

The  identity check problem is concerned with
estimating the distance $\delta(U)$.
Checking an approximate equivalence of $n$-qubit quantum circuits $U_1$ and $U_2$ is a special
case of this problem since the diamond-norm distance between 
$U_1$ and $U_2$  coincides with $\delta(U_2^\dag U_1)$. 
An identity check  algorithm   is said to achieve an approximation ratio $\alpha\ge 1$
for a class of quantum circuits $\calC$ if it takes as input a circuit $U\in \calC$
and outputs a real number $\gamma$ such that 
\be
\label{alpha}
\delta(U) \le \gamma \le \alpha \delta(U)
\ee
for all circuits $U\in \calC$.
The algorithm is  efficient if its runtime scales at most polynomially with the number of qubits $n$
for a fixed approximation ratio $\alpha$.

Our main result is a classical identity check algorithm for shallow  geometrically local circuits.
We assume that $n$ qubits are located at cells of a $D$-dimensional
rectangular array and consider circuits
composed of single-qubit
and two-qubit gates acting on nearest-neighbors cells (cells $i$ and $j$ are called nearest-neighbors if one can go from $i$ to $j$ by 
changing a single coordinate by $\pm1$).
A depth-$h$ circuit consists of $h$ layers of gates such that 
within each layer all gates are disjoint.
Our identity check algorithm for $D$-dimensional circuits
achieves an approximation ratio 
\be
\label{alpha_diamond_norm}
\alpha=D+1
\ee
if the input circuit satisfies $\delta(U)<2$
and
$\alpha=1.16(D+1)$ in the general case.
The runtime of the algorithm is
\be
\label{runtime_D}
T\sim n 2^{12(2hD)^D}.
\ee
The runtime is linear in $n$ for any constant circuit depth $h$
and spatial dimension $D$. 
We note that achieving an approximation ratio $\alpha=1+\epsilon$ with $\epsilon=poly(1/n)$ is at least as hard as approximating the distance
$\delta(U)$ with an additive error $poly(1/n)$. The latter problem is known to be QMA-hard even in the case
of constant-depth 1D circuits~\cite{ji2009non} which rules out efficient algorithms.
An interesting open problem is whether an efficient classical or quantum algorithm can obtain
an approximation $\alpha=1+\epsilon$ for any constant $\epsilon>0$. If true, this would
provide a Polynomial Time Approximation Scheme~\cite{vazirani2001approximation} for the identity check problem.

 Applications such as Quantum Phase Estimation~\cite{nielsen2010quantum}
or Krylov subspace algorithms~\cite{huggins2020non,seki2021quantum,kirby2023exact}
are sensitive to the overall phase of a quantum circuit since the circuit may be controlled by ancillary qubits.
This motivates
a phase-sensitive version of the identity check problem where the goal is
to estimate the operator-norm distance $\|U-I\|$, that is, the largest singular value of
$U-I$.  As before, we aim at approximating $\|U-I\|$ with a constant multiplicative error.

A natural strategy
is to reduce the task
of approximating $\|U-I\|$ to the one of approximating the diamond-norm distance $\delta(U)$, which has been already
addressed.
It is clear however that such a reduction may not always be possible. For example, if $U=e^{i\varphi} I$ is a multiple
of the identity then $\delta(U)=0$ while $\|U-I\|$ may take any value between $0$ and $2$.
To overcome this obstacle, our algorithm requires
an additional input data
which depends on the phase of $U$. Namely, let  $P_U$ be the smallest convex subset of 
the complex plane $\CC^2$ that contains all eigenvalues of $U$. 
Equivalently, $P_U$ is a polygon whose
vertices are eigenvalues of $U$. Since $U$ is unitary, all vertices of $P_U$ lie on the unit circle. 
It is known~\cite{aharonov1998quantum} that the polygon $P_U$ provides a simple geometric interpretation of the diamond-norm distance, see Fig.~\ref{fig:1}.

\begin{figure}[h]
         \includegraphics[width=0.25\textwidth]{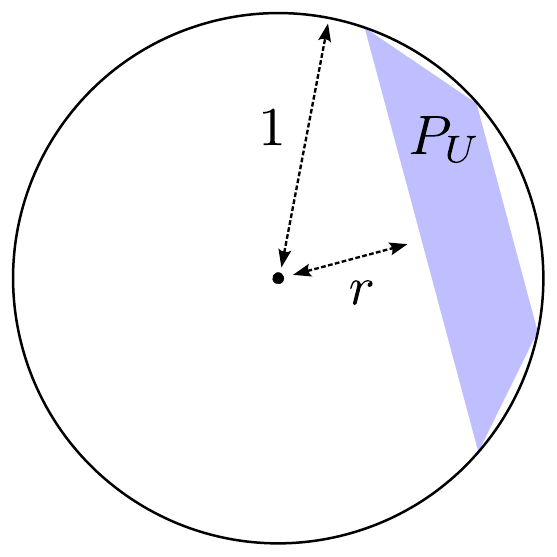}
     \caption{Eigenvalue polygon $P_U$ whose vertices are eigenvalues of $U$. The diamond-norm distance between $U$ and the identity channel is $\delta(U)=2\sqrt{1-r^2}$, where $r$ is the distance between $P_U$ and the origin~\cite{aharonov1998quantum}. If $P_U$ does not contain the origin then $\delta(U)$ coincides with the diameter of $P_U$. Otherwise, $\delta(U)=2$.}  
     \label{fig:1}
 \end{figure}

Our approximation algorithm for the phase-sensitive identity check problem takes as input a $D$-dimensional depth-$h$
circuit acting on $n$ qubits and an arbitrary point $t\in P_U$.
The algorithm outputs an estimator 
\be
\label{gamma_op}
\gamma_{op} = \gamma + |t-1|,
\ee
where $\gamma$ is an estimate of the diamond-norm distance $\delta(U)$ 
satisfying $\delta(U)\le \gamma \le \alpha \delta(U)$
obtained by calling the identity check algorithm for the diamond-norm $\delta(U)$.
We show that 
\be
\label{alpha_op}
\|U-I\| \le \gamma_{op} \le \alpha_{op}  \|U-I\|.
\ee
with
\be
\label{alpha_op'}
\alpha_{op} = 1+2\alpha.
\ee
Suppose $\rho$ is some $n$-qubit state such that the trace $\mathrm{Tr}(\rho U)$ can be computed efficiently. Note that $\mathrm{Tr}(\rho U)\in P_U$ since the diagonal of $\rho$ in the eigenbasis of $U$ is a probability distribution. Thus one can use the estimator Eq.~(\ref{gamma_op}) with $t=\mathrm{Tr}(\rho U)$. For example, if $U$ is a 1D shallow circuit, one can choose $\rho$ as an arbitrary product state. Since $U$ is a Matrix Product Operator with a bond dimension $2^{O(h)}$ one can compute $\mathrm{Tr}(\rho U)$ efficiently using algorithms based on 
Matrix Product States~\cite{schollwock2011density} as long as $h=O(\log{n})$. 
In the 1D case Eqs.~(\ref{alpha_diamond_norm},\ref{alpha_op'}) give $\alpha=2$ and
$\alpha_{op}=5$ while the runtime of the algorithm is $T\sim n2^{O(h)}$, see Eq.~(\ref{runtime_D}).
As another example,
suppose $U$ is a Trotter circuit describing time evolution of a $D$-dimensional Hamiltonian 
composed of local Pauli  terms $XX+YY$, $ZZ$, and $Z$ that preserve the Hamming weight.
Then the all-zeros state $|0^n\ra$
is a common eigenvector of each 
individual gate in $U$  and one can choose
$\rho$ as the all-zeros state, that is,
$t=\la 0^n|U|0^n\ra$. 
From Eqs.~(\ref{alpha_diamond_norm},\ref{alpha_op'}) one gets $\alpha_{op}=2D+3$.
In general, the above gives an efficient algorithm approximating $\|U-I\|$ within a factor
$\alpha_{op}=2D+3$ for $D$-dimensional constant-depth circuits provided that one can efficiently find at least
one point in the eigenvalue polygon $P_U$.

The unfavorable runtime scaling of our algorithm with the circuit depth limits its application to very shallow circuits.
However, the algorithm can be extended to deep circuits $U$ using the divide and conquer strategy.
Namely, if $U=U_\ell \cdots U_2 U_1$ where each layer $U_i$
has depth $O(1)$, the triangle inequality gives $\delta(U)\le \sum_{i=1}^\ell \delta(U_i) \le \sum_{i=1}^\ell \gamma_i$, where
$\gamma_i$ is an upper bound on $\delta(U_i)$ computed by our algorithm.
The runtime for computing this upper bound on $\delta(U)$ scales only linearly with the depth of $U$ but we can no longer guarantee that the upper
bound is tight within a constant factor. Other tradeoffs between the runtime and the upper bound tightness are discussed in Section~\ref{sec:algo1}.

Although this work primarily  focuses on computing upper bounds on the distance to the identity, as required for validation of quantum algorithms, efficiently computable lower bounds on the distance are also of
interest. Density Matrix Renormalization Group (DMRG) algorithms~\cite{schollwock2011density}  provide a powerful tool for 
computing lower bounds on the distance $\delta(U)$ or $\|U-I\|$ for 1D shallow circuits $U$. 
Indeed, one can easily check that  the squared distance $\|U-I\|^2$ coincides with the largest eigenvalue
of a Hamiltonian $H=2I-U-U^\dag$. If $U$ is a depth-$h$ 1D circuit then $H$
is a Matrix Product Operator (MPO) with 
a bond dimension $2^{O(h)}$. In practice, extremal eigenvalues of MPO Hamiltonians with a small bond dimension 
can be well approximated using DMRG algorithms~\cite{schollwock2011density}. 
However, since DMRG is a variational algorithm, it only provides a lower bound on the
distance $\|U-I\|$. To lower bound the diamond-norm distance we use a bound
\[
\delta(U)\ge \|U\otimes U^\dag - I\otimes I\|,
\]
with the equality if $\delta(U)<2$,
see Section~\ref{sec:commutators}. Thus $\delta(U)^2$ is lower bounded by the maximum eigenvalue
of an MPO Hamiltonian $H=2I\otimes I - U\otimes U^\dag - U^\dag \otimes U$ which can in turn be lower bounded using DMRG
algorithm. We leave the study of lower bounds based on DMRG algorithms for a future work.

The rest of the paper is organized as follows. 
Section~\ref{sec:commutators} describes bounds on the diamond-norm and operator-norm distances 
$\delta(U)$ and $\|U-I\|$
that can be expressed in terms of commutators between $U$ and certain observables. 
This section also sketches main ideas behind our algorithm. 
Section~\ref{sec:lightcones} 
collects some basic facts about shallow quantum circuits and 
$D$-dimensional partitions.
Section~\ref{sec:additivity} proves a technical lemma which
relates the norms of global and local commutators. 
Our identity check algorithm and its analysis is presented in Section~\ref{sec:algo1}.
Finally, Section~\ref{sec:numerics} reports a software implementation of our algorithm.

\section{Commutator-based bounds }
\label{sec:commutators}

Our identity check algorithm borrows many ideas from the recent breakthrough
work by Huang, Liu, et al.~\cite{learning2024} on learning shallow quantum circuits.
The main ingredients of our algorithm, described below, are bounds on the diamond-norm
distance $\delta(U)$ that  depend on the
norm of commutators between $U$ and certain observables composed
of SWAP gates. These bounds and their proof are largely based on Ref.~\cite{learning2024}.

Consider $2n$ qubits labeled by integers $1,\ldots,2n$.
Let $W_i$ be the SWAP gate applied to qubits $i$ and $i+n$.
Given a subset $A\subseteq [n]$, define a $2n$-qubit operator
\[
W_A =\prod_{i\in A} W_i.
\]
By definition, $W_A$ acts non-trivially on $2|A|$ qubits.
\begin{lemma}
\label{lemma:comm1}
Let $[n]=A_1\ldots A_m$ be a partition of $n$ qubits into $m$ disjoint subsets
and $U$ be a unitary operator acting on $n$ qubits.
Define a quantity 
\be
\label{gamma_eq1}
\gamma =  \sum_{j=1}^m \|  W_{A_j} (U \otimes I) W_{A_j}  (U^\dag \otimes I) - I\otimes I\|.
\ee
Then 
\be 
\label{gamma_eq2}
\delta(U)  \le \gamma \le m\delta(U)
\ee
assuming that  $\delta(U)<2$ and 
\begin{align} 
\label{gamma_eq3}
\delta(U)  \le 1.16\gamma \le 1.16 m\delta(U)
\end{align}
in the general case. 
\end{lemma}
The quantity $\gamma$ defined in Eq.~(\ref{gamma_eq1}) 
or its rescaled version 
$1.16\gamma$
will be the desired estimator of 
the distance $\delta(U)$.
In the next section we show how to 
choose a partition $[n]=A_1\ldots A_m$ with $m=D+1$
parts  such that 
each subset $A_j$ is a union of well-separated hypercubes of linear size $O(hD)$
and all commutators 
 $W_{A_j} (U \otimes I) W_{A_j}  (U^\dag \otimes I)$
 that appear in Eq.~(\ref{gamma_eq1}) are tensor products of local commutators
supported on individual hypercubes. Our construction is based on Ref.~\cite{woude2022geometry}
which introduced  so-called
reclusive partitions of the $D$-dimensional Euclidean space. 
The key ingredient of our algorithm is an additivity lemma stated in Section~\ref{sec:additivity}.
This lemma expresses the norm of commutators
$\|  W_{A_j} (U \otimes I) W_{A_j}  (U^\dag \otimes I) - I\otimes I\|$
in terms of the norm of analogous local commutators supported on individual hypercubes. 
Each local commutator acts on a subset of at most $O(hD)^D$ qubits
and its eigenvalues can be computed by the exact diagonalization.
 The additivity lemma then
provides a linear time algorithm for computing the norm of global commutators $\|  W_{A_j} (U \otimes I) W_{A_j}  (U^\dag \otimes I) - I\otimes I\|$ which is all we need
to compute the estimator $\gamma$ defined in Lemma~\ref{lemma:comm1}.

The next lemma shows that estimation of the operator-norm distance
can be reduced to estimation of diamond-norm distance given any point in the eigenvalue polygon of $U$.
\begin{lemma}
\label{lemma:reduction}
Let $t\in P_U$ be any point in the eigenvalue polygon of $U$ and $\alpha,\gamma$ be  real numbers such that 
$\delta(U)\le \gamma \le \alpha \delta(U)$. Then
\[
\gamma_{op} = \gamma + |t-1|
\]
obeys
\[
\|U-I\|\le \gamma_{op}\le (1+2\alpha) \|U-I\|.
\]
\end{lemma}
In the rest of this section we prove Lemma~\ref{lemma:comm1} and \ref{lemma:reduction}.

\begin{proof}[\bf Proof of Lemma~\ref{lemma:comm1}]
Consider first the case $\delta(U)<2$. We claim that in this case
\be
\label{sec1eq1}
\delta(U)= \| U\otimes U^\dag - I\otimes I\|.
\ee
Indeed, since $\delta(U)<2$,
the eigenvalue polygon $P_U$ does not contain the origin and thus $\delta(U)$ coincides with the diameter of $P_U$, see Fig.~\ref{fig:1}.
Let $\{e^{i\varphi_a}\}_a$ be eigenvalues of $U$. By definition, $P_U$ is the convex hull of points $\{e^{i\varphi_a}\}_a$. Hence the diameter of $P_U$ coincides with the maximum distance between eigenvalues of $U$. This shows that 
\begin{align*}
\delta(U)& =\mathrm{diam}(P_U)
=\max_{a,b}|e^{i\varphi_a} - e^{i\varphi_b}| \\
&=\max_{a,b}|e^{i(\varphi_a-\varphi_b)} - 1| \\
& = \|U\otimes U^\dag - I\otimes I\|.
\end{align*}
To get the last equality we noted that $\{e^{i(\varphi_a-\varphi_b)}-1\}_{a,b}$ is the set of eigenvalues of $U\otimes U^\dag-I\otimes I$.

Let us agree that the tensor product in Eq.~(\ref{sec1eq1}) separates two $n$-qubit registers
that span qubits $\{1,\ldots,n\}$ and $\{n+1,\ldots,2n\}$.
Let $W=\prod_{i=1}^n W_i$ be an operator that swaps the two registers.
Since the operator norm is unitarily invariant, 
Eq.~(\ref{sec1eq1}) gives
\begin{align}
\label{sec1eq2}
\delta(U) & = \| (U\otimes U^\dag -  I\otimes I)W\| \nonumber \\
& = \| (U\otimes I) W (U^\dag  \otimes I) - W\|.
\end{align}
Here we noted that $(I\otimes U^\dag) W  = W (U^\dag \otimes I)$.
The triangle inequality implies that 
for any unitary operators $P_j,Q_j$ one has 
\be
\label{triangle_inequality}
\| P_1 P_2 \cdots P_m - Q_1 Q_2 \cdots Q_m \| \le \sum_{j=1}^m \| P_j - Q_j\|.
\ee
Choosing $P_j = (U\otimes I) W_{A_j}  (U^\dag \otimes I)$, $Q_j = W_{A_j}$,
and noting that $W=\prod_{j=1}^m W_{A_j}$ one arrives at
\be
\label{sec1eq3}
\delta(U) \le \sum_{j=1}^m \| (U \otimes I) W_{A_j}  (U^\dag  \otimes I) - W_{A_j}\|=\gamma.
\ee
The last equality uses the fact that $W_{A_j}$ are both hermitian and unitary, which implies
$\|O-W_{A_j}\|=\|W_{A_j} O-I\|$ for any operator $O$.
The dual characterization of the diamond-norm~\cite{watrous2009semidefinite} gives
\be
\label{sec1eq4}
\delta(U)=\max_{V\, : \, \|V\|\le 1}\; \| (U \otimes I) V (U^\dag \otimes I) - V\|
\ee
where the maximization is over $2n$-qubit operators $V$.
Since $\| W_{A_j}\|=1$ one infers that
\[
 \| (U \otimes I) W_{A_j}  (U^\dag  \otimes I) - W_{A_j}\|\le \delta(U)
\]
for all $j$ and thus 
$\gamma \le m \delta(U)$. This concludes the proof in the case $\delta(U)<2$.

Suppose now  that $\delta(U)=2$. Then 
the eigenvalue polygon $P_U$ contains the origin, see Fig.~\ref{fig:1}. Let $\{e^{i\varphi_a}\}_a$ be the eigenvalues of $U$.
We claim that there exist eigenvalues $e^{i\varphi_0}, e^{i\varphi_1}$ of $U$ such that the shortest arc length between them is at least $2\pi/3$. Otherwise, all eigenvalues would lie within an arc of length $2\pi/3$, 1/3 of the unit circle --- but this would imply that $P_U$ does not contain the origin.
Thus
\begin{align}
    \| U\otimes U^\dag - I\otimes I\| & = \max_{a,b} |e^{i(\varphi_a-\varphi_b)}-1|  \\
    &\geq 
    |e^{i(\varphi_0-\varphi_1)}-1|\\
    &\ge
    |e^{i2\pi/3}-1| 
    \\
    &= 2\sin\left(\pi/3\right) = \sqrt{3}.
\end{align}
Therefore we have
\begin{align}
    \gamma \geq \| U\otimes U^\dag - I\otimes I\| \geq \sqrt{3}
\end{align}
so
\begin{align}
    \frac{2}{\sqrt{3}}\gamma \geq 2 = \delta(U).
\end{align}
Furthermore, our proof of the upper bound $\gamma \leq m \delta(U)$ is unchanged when $\delta(U)=2$. The desired bound, Eq.~(\ref{gamma_eq3}) follows since $1.16 \geq \frac{2}{\sqrt{3}}$.
\end{proof}

\begin{proof}[\bf Proof of Lemma~\ref{lemma:reduction}]

Let $\{e^{i\varphi_a}\}_a$ be eigenvalues of $U$ and $t=\sum_a p_a e^{i\varphi_a}$, where
$p_a\ge 0$ and $\sum_a p_a=1$. We have 
\begin{align*}
\|U-I\| & = \|U-tI+tI-I\|\\
& \le |t-1| + \| \sum_a p_a (U - e^{i\varphi_a} I)\|\\
& \le |t-1| + \sum_a p_a \| U- e^{i\varphi_a} I\| \\
& \le |t-1| + \max_a \|U - e^{i\varphi_a}I\| \\
& = |t-1| + \max_{a,b} |e^{i\varphi_a} - e^{i\varphi_b}| \\
& \le  |t-1|+\delta(U) \le |t-1| + \gamma.
\end{align*}
Conversely, it is well known~\cite{aharonov1998quantum} that $\delta(U)\le 2\|U-I\|$ for any untary $U$. Thus
\begin{align*}
|t-1| +\gamma & =
\left| \sum_a p_a (e^{i\varphi_a}-1)\right| + \gamma \\
& \le  \sum_a p_a |e^{i\varphi_a} - 1| + \alpha \delta(U)\\
& \le \max_a |e^{i\varphi_a} - 1| + 2\alpha \|U-I\| \\
&=(1+2\alpha)\|U-I\|.
\end{align*}

\end{proof}

\section{Lightcones and reclusive partitions}
\label{sec:lightcones}

Given a quantum circuit $U$ acting on $n$ qubits, 
the lightcone $\calL(j)$ of a qubit $j\in [n]$ is defined as the set of all output qubits $i\in [n]$
that can be reached by moving through the circuit diagram forward in time starting from the input qubit $j$.
For example, if $U$ is  a one-dimensional  circuit of depth $h$  then 
\be
\label{lightcone1}
\calL(j)\subseteq [j-h,j+h].
\ee
For any subset of qubits $S\subseteq [n]$ let $\calL(S)$ be the lightcone of $S$ defined as
\be
\label{lightcone2}
\calL(S)=\bigcup_{j\in S} \calL(j).
\ee
We say that a subset of qubits  $S$ is the support of an operator $O$ and write $S=\mathrm{supp}(O)$ if
$O$ acts trivially on all qubits $j\notin S$. By definition, 
\be
\mathrm{supp}(UOU^\dag) \subseteq \calL(\mathrm{supp}(O))
\ee
for any operator $O$. Furthermore, $UOU^\dag = U_{loc} O U_{loc}^\dag$, where
$U_{loc}$ is a "localized" circuit obtained from $U$ by removing all gates
acting on qubits outside of the lightcone $\calL(\mathrm{supp}(O))$.

Two subsets of qubits $S_1$ and $S_2$ are said to be lightcone separated if $\calL(S_1)\cap \calL(S_2)=\emptyset$.
If $O_1$ and $O_2$ are operators supported on $S_1$ and $S_2$ then 
$UO_1 O_2 U^\dag$ is a product of operators $UO_1U^\dag$ and $UO_2 U^\dag$ with disjoint supports.

\begin{figure}[t]
         \includegraphics[width=0.4\textwidth]{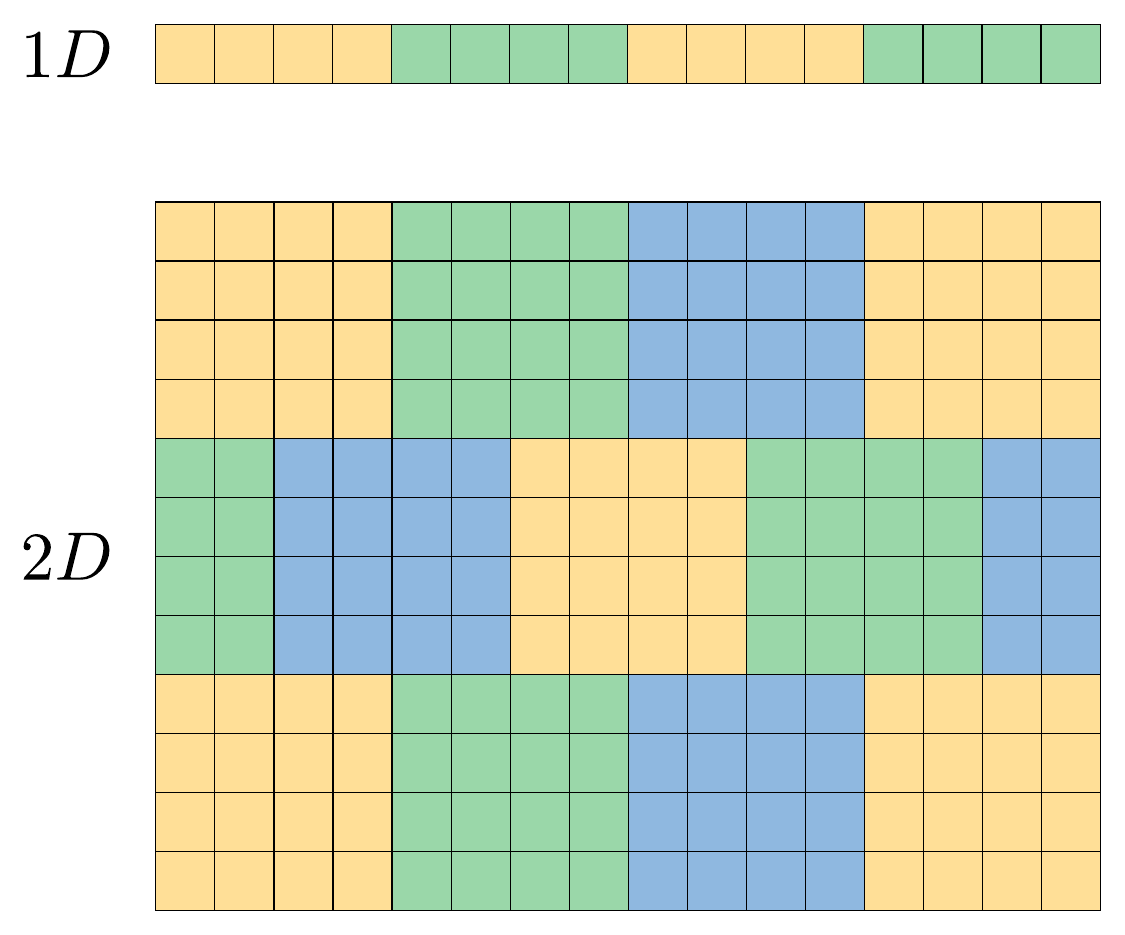}
     \caption{Examples of reclusive partitions for $D=1,2$.
     Qubits are located at cells of a $D$-dimensional rectangular array.
     The array is partitioned into $D+1$ sets $A_1,\ldots,A_{D+1}$ such that 
     each set $A_j$  is a disjoint union of  $D$-dimensional cubes of linear size $L$ 
     and the 
     distance between any pair of cubes from the same set $A_j$ is at least $L/D$.
     Here $L=4$.
     Cubes located near the boundary of the array are truncated.
     The sets $A_1,A_2,A_3$ are highlighted in yellow, green, and blue.}  
     \label{fig:partition}
 \end{figure}

Suppose now that $n$ qubits are located at cells of a $D$-dimensional rectangular array.
We shall consider partitions of the array into $D$-dimensional cubes
known as reclusive partitions~\cite{woude2022geometry}.
The linear size of each cube in the partition will be chosen as
\be
\label{cube_size}
L=2Dh,
\ee 
where $h$ is the depth of $U$.
\begin{lemma}[\bf Reclusive Partitions~\cite{woude2022geometry}]
\label{lemma:partition}
One can partition cells of a $D$-dimensional rectangular array  into $D+1$ 
sets $A_1,\ldots,A_{D+1}$ such that each set $A_j$
is a disjoint union of $D$-dimensional cubes of linear size $L$
and the distance between any pair of cubes from the same set $A_j$ is at least $L/D$.
The above partition can be constructed efficiently.
\end{lemma}
Figure~\ref{fig:partition} shows examples of 1D and 2D reclusive partitions,
see Ref.~\cite{woude2022geometry} for the 3D example.
We defer the proof
of Lemma~\ref{lemma:partition} to Appendix~\ref{app:partition} since it is a simple rephrasing
of the results established in~\cite{woude2022geometry}.
By construction, 
each cube in the partition  contains at most $L^D$ qubits (cubes located near the boundary of the array may be truncated)
and
any pair of cubes from the same set $A_j$ is lightcone separated due to Eq.~(\ref{cube_size}).
Write
\[
A_j = A_{j,1} A_{j,2} \ldots A_{j,\ell_j},
\]
where $\ell_j$ is the number of cubes in $A_j$ 
and $A_{j,p}$ denotes the $p$-th cube in $A_j$. By constriction, we have
\be
\label{LCseparation}
\calL(A_{j,p}) \cap \calL(A_{j,q})=\emptyset \quad \mbox{for all $p\ne q$}.
\ee
Since the lightcone of a cube with a linear size $L$ can be enclosed by a cube
of linear size $L+2h$, the number of qubits contained in any lightcone $\calL(A_{j,p})$ is bounded as
\be
\label{lightcone_size}
|\calL(A_{j,p})| \le  (2h(D+1))^D.
\ee
Here we used Eq.~(\ref{cube_size}).

Consider the diamond-norm distance $\delta(U)$ and specialize 
the commutator-based bound
of Lemma~\ref{lemma:comm1} to the reclusive partition $[n]=A_1\ldots A_{D+1}$.
By definition, 
\[
W_{A_j} = \prod_{p=1}^{\ell_j} W_{A_{j,p}}.
\]
Lightcone separation of cubes $A_{j,p}$ implies that operators
$(U\otimes I) W_{A_{j,p}} (U^\dag \otimes I)$ acts on pairwise disjoint subsets of qubits.
Thus
\be
W_{A_j} (U \otimes I) W_{A_j}  (U^\dag \otimes I) =\prod_{p=1}^{\ell_j} K_{j,p},
\ee
where we defined commutators
\[
K_{j,p} = W_{A_{j,p}} (U \otimes I) W_{A_{j,p}}  (U^\dag \otimes I).
\]
The above shows that $K_{j,p}$ are operators acting on pairwise disjoint subsets of qubits (for a fixed $j$).
Let $U_{j,p}$ be a "localized" circuit obtained from $U$ by replacing all gates acting on at least one qubit
outside of the lightcone $\calL(A_{j,p})$ with the identity. Then $U_{j,p}$ acts non-trivially only
on the lightcone $\calL(A_{j,p})$ and
\[
K_{j,p} = W_{A_{j,p}} (U_{j,p} \otimes I) W_{A_{j,p}}  (U_{j,p}^\dag \otimes I).
\]
The support of $K_{j,p}$  includes all qubits in the left $n$-qubit register contained in $\calL(A_{j,p})$ 
as well as all qubits in the right $n$-qubit register contained in $A_{j,p}$. Thus
\begin{align*}
|\mathrm{supp}(K_{j,p})| & \le |\calL(A_{j,p})| + |A_{j,p}| \\
&  \le  (2h(D+1))^D + (2hD)^D \\
& = (2hD)^D \left[(1+1/D)^D +1 \right] \le 4(2hD)^D.
\end{align*}
Eigenvalues of a unitary operator acting on $m$ qubits can be computed in time $O(2^{3m})$ by 
the exact diagonalization of a unitary $2^m\times 2^m$ matrix. 
Thus one can compute all eigenvalues of the commutator $K_{j,p}$ 
 in time 
\[
T\sim 2^{12 (2hD)^D}.
\]
In the next section we show that the norm
\[
\| W_{A_j} (U \otimes I) W_{A_j}  (U^\dag \otimes I) -I\otimes I\|=\| \prod_{p=1}^{\ell_j} K_{j,p} - I\otimes I\|
\]
 that appears in the bound
of Lemma~\ref{lemma:comm1} is a simple function 
of  eigenvalues of individual commutators $K_{j,p}$.

\section{Additivity lemma}
\label{sec:additivity}

In this section we show how to compute the norm
of commutators
that appear in Lemma~\ref{lemma:comm1}.
First, let us introduce some terminology.
Let $S^1=\{ z\in \CC\, : \, |z|=1\}$ be the unit circle. If $U$ is a unitary operator, let  $\mathsf{eig}(U)\subseteq S^1$
be the set of eigenvalues of $U$ (ignoring multiplicities).
Consider $2n$ qubits, a subset $A\subseteq [n]$,  and a SWAP operator
$W_A=\prod_{i\in A} W_i$ where $W_i$ is the SWAP gate acting on qubits $i$ and $i+n$.
Consider a commutator
\[
K_A=W_A (U\otimes I) W_A (U^\dag \otimes I).
\]
We claim that $\mathsf{eig}(K_{A})=\mathsf{eig}(K_{A}^\dag)$.
Indeed, $K_{A}^\dag = W_A K_{A} W_A$.
Since $W_A$ is both unitary and hermitian, conjugation by $W_A$ does not change the eigenvalue spectrum.
Thus eigenvalues of $K_{A}$ have a form $e^{\pm i\varphi}$ with $0\le \varphi\le \pi$.
For each $\varphi$ one can choose both positive and negative sign in the exponent.
 Define a function $\theta$ that maps subsets of qubits $A\subseteq [n]$ to
 real numbers in the interval $[0,\pi]$ such that 
\be
\theta(A)=\max_{\varphi\in [0,\pi]} \varphi \quad \mbox{subject to} \quad e^{i\varphi} \in \mathsf{eig}(K_A).
\ee
Note that  $e^{i\theta(A)}$ is the unique eigenvalue
of $K_A$ with  the maximum distance from $1$
and a non-negative imaginary part.
Accordingly, 
\be
\label{theta_vs_distance}
\|K_A -I\| = |e^{i\theta(A)}-1|.
\ee
We shall need the following simple fact. 
\begin{lemma}
\label{lemma:badness}
If $\theta(A)\ge \pi/2$ for some subset $A\subseteq [n]$
then  $\delta(U)\ge \sqrt{2}$.
\end{lemma}
\begin{proof}
From $\theta(A)\ge \pi/2$ one infers that 
 $K_A$
has an eigenvalue  with a non-positive real part. 
Since  all points on the unit circle within distance less than $\sqrt{2}$ from $1$ have a positive real part,
one gets $\|K_A-I\| \ge \sqrt{2}$. The dual characterization of the diamond norm~\cite{watrous2009semidefinite} gives
\begin{align*}
\delta(U)& =\max_{\eta\, : \, \|\eta\|\le 1} \; \| (U\otimes I) \eta (U^\dag \otimes I) - \eta \| \\
& \ge \| (U\otimes I) W_A (U^\dag \otimes I) - W_A\| = \|K_A-I\| \ge \sqrt{2}.
\end{align*}
\end{proof}
\begin{dfn}
A subset $A\subseteq [n]$  is called good if $\theta(A)<\pi/2$.
Otherwise $A$ is called bad. 
\end{dfn}
The following lemma shows that the function $\theta(A)$ is additive under the union of  lightcone-separated subsets, provided that the circuit $U$ is sufficiently close to the identity.
\begin{lemma}[\bf Additivity]
Suppose  $A_1,A_2\subseteq [n]$ are good lightcone-separated subsets.
Consider two cases:
\begin{enumerate}
\item[(a)] $\theta(A_1)+\theta(A_2)<\pi/2$,
\item[(b)] $\theta(A_1)+\theta(A_2)\ge \pi/2$.
\end{enumerate}
Case (a) implies that the union $A_1 A_2$ is good and 
\be
\label{theta_additive}
\theta(A_1A_2)=\theta(A_1)+\theta(A_2).
\ee
Case (b) implies that $\delta(U)\ge \sqrt{2}$.
\end{lemma}
\begin{proof}
Define commutators
\[
K_p=W_{A_p} (U\otimes I) W_{A_p} (U^\dag \otimes I)
\]
with $p\in \{1,2\}$.
Since $A_1$ and $A_2$ have lightcone separated,
$K_1$ and $K_2$ act on disjoint subsets of qubits and thus 
\[
K_{12} \equiv  W_{A_1 A_2}  (U\otimes I) W_{A_1 A_2}  (U^\dag \otimes I) =K_1K_2
\]
has the same eigenvalues as the tensor product of $K_1$ and $K_2$. In other  words,
\[
\mathsf{eig}(K_1K_2) = \{z_1z_2\, : \, z_1\in \mathsf{eig}(K_1) \quad \mbox{and} \quad z_2\in \mathsf{eig}(K_2)\}.
\]
By definition, $e^{i\theta(A_p)}\in  \mathsf{eig}(K_p)$ for $p=1,2$.
Thus  $e^{i\theta(A_1)+ i\theta(A_2)}\in \mathsf{eig}(K_1K_2)= \mathsf{eig}(K_{12})$.

Consider case (a). Let  $e^{i\varphi_p}\in  \mathsf{eig}(K_p)$ be eigenvalues such that 
$e^{i\theta(A_1A_2)} = e^{i(\varphi_1+\varphi_2)}$.
Then 
\be
\label{case_a}
\theta(A_1 A_2)= \varphi_1 + \varphi_2 + 2\pi k
\ee
for some integer $k$ chosen such that $\theta(\sigma_1 \sigma_2)\in [0,\pi]$.
By definition, $|\varphi_p|\le \theta(A_p)$ and thus 
\[
|\varphi_1|+|\varphi_2|\le \theta(A_1)+\theta(A_2)<\frac{\pi}2.
\]
Hence  the only integer  $k$ in  Eq.~(\ref{case_a}) satisfying 
$\theta(A_1 A_2)\in [0,\pi]$ is $k=0$, that is,
$\theta(A_1A_2) = \varphi_1+\varphi_2\le \theta(A_1)+\theta(A_2)$.
Conversely,  since $e^{i\theta(A_1)+ i\theta(A_2)}$ is an eigenvalue of $K_{12}$
and $\theta(A_1)+\theta(A_2)<\pi/2$, one infers that 
$\theta(A_1A_2)\ge \theta(A_1)+\theta(A_2)$.
This proves Eq.~(\ref{theta_additive}).

Consider case (b). The same arguments as above show that $K_{12}$ has an eigenvalue
$e^{i\varphi}$, where $\varphi = \theta(A_1)+\theta(A_2)\in [\pi/2,\pi)$.
Here we used the assumption that both $A_1$ and $A_2$ are good,
as well as the bound $\theta(A_1)+\theta(A_2)\ge \pi/2$.
Hence $\theta(A_1A_2)\ge \pi/2$ and $\delta(U)\ge \sqrt{2}$
by
Lemma~\ref{lemma:badness}.
\end{proof}
By inductive application of the additivity lemma one obtains the following.
\begin{corol}
\label{corol:1}
Suppose $A_1,\ldots,A_\ell \subseteq [n]$ are lightcone separated subsets.
Let $A=\cup_{p=1}^\ell A_p$ be their union and
\be
\label{sum_over_cubes}
\varphi =\sum_{p=1}^{\ell} \theta(A_p).
\ee
Here the angles are added as real numbers (rather than modulo $2\pi$).
If $\varphi<\pi/2$ then
\be
\| W_A (U\otimes I) W_A (U^\dag \otimes I)  - I\| = |e^{i\varphi}-1|.
\ee
If  $\varphi \ge \pi/2$ then $\delta(U)\ge \sqrt{2}$.
\end{corol}

\section{Identity check algorithm }
\label{sec:algo1}

Combining all above ingredients we arrive at the following algorithm
for the $D$-dimensional identity check problem.
We first consider the case when the input circuit $U$ is sufficiently close to the identity such that $\delta(U)<2$.
Below we assume that a reclusive partition $[n]=A_1\ldots A_{D+1}$ of the $D$-dimensional qubit array 
has been already computed, see Appendix~\ref{app:partition} for details. 
We claim that the following algorithm outputs an estimator $\gamma$
satisfying $\delta(U)\le \gamma \le (D+1)\delta(U)$.
\begin{algorithm}[H]
	\caption{Identity check (diamond-norm)\label{algo1}}
\textbf{Input:}  An $n$-qubit $D$-dimensional circuit $U$ with $\delta(U)<2$.\\
 \textbf{Output:} $\gamma \in \RR$ satisfying  $\delta(U)\le \gamma \le (D+1)\delta(U)$.
	\begin{algorithmic}[1]
\State{$\gamma\gets 0$}
\For{$j=1$ to $D+1$}
\State{$\varphi_j\gets 0$}
\State{$\ell_j\gets$ number of cubes in $A_j$}
\For{$p=1$ to $\ell_j$}
\State{$A_{j,p}\gets$ $p$-th cube in $A_j$}
\State{$\varphi_j\gets \varphi_j + \theta(A_{j,p})$}
\If{$\varphi_j\ge \pi/2$}
\State{\Return $\gamma=2$}
\EndIf
\EndFor
\State{$\gamma\gets \gamma + |e^{i\varphi_j}-1|$}
\EndFor	
\end{algorithmic}
\end{algorithm}
Indeed, if line~9 is never reached, 
Corollary~\ref{corol:1} of the additivity lemma
imply that the output of the algorithm coincides with the quantity $\gamma$
defined in Lemma~\ref{lemma:comm1}
specialized to the reclusive partition.
In this case correctness of the algorithm follows directly
from Lemma~\ref{lemma:comm1}.
Otherwise, the algorithm outputs
$\gamma=2$, while Corollary~\ref{corol:1}
implies that $\delta(U)\ge \sqrt{2}$. In this case $\gamma=2$
satisfies the bounds $\delta(U)\le \gamma \le (D+1)\delta(U)$ for $D\ge 1$.
We claim that the algorithm runs in time  $O(n 2^{12(2hD)^D})$.
Indeed, the total number of 
cubes $A_{j,p}$ is $O(n)$. Computing the function $\theta(A_{j,p})$ at line~7
requires eigenvalues of a unitary operator 
$K_{A_{j,p}}$ acting on at most $4(2hD)^D$ qubits,
as discussed in Section~\ref{sec:lightcones}.
This computation
takes time $O(2^{12 (2hD)^D})$. Hence the total runtime is $O(n 2^{12(2hD)^D})$.

Next consider the general case when it is possible that $\delta(U)=2$. 
Define our estimator of $\delta(U)$  as $1.16\gamma$, where $\gamma$ is the output of Algorithm~1.
We claim that 
\be
\label{general_case}
\delta(U)\le 1.16 \gamma \le 1.16 (D+1)\delta(U).
\ee
If the algorithm never reaches line~9 then 
its output coincides with the quantity $\gamma$
defined in Lemma~\ref{lemma:comm1}
and Eq.~(\ref{general_case})
follows directly from Lemma~\ref{lemma:comm1},
see Eq.~(\ref{gamma_eq2}).
Otherwise, if the algorithm reaches line~9, it
outputs $\gamma=2$ while $\delta(U)\ge \sqrt{2}$
due to Corollary~\ref{corol:1} of the additivity lemma.
In this case the first inequality in
Eq.~(\ref{general_case}) 
follows from $\delta(U)\le 2$
and the second inequality
becomes
$2\le (D+1)\delta(U)$ which is true for any $D\ge 1$
since $\delta(U)\ge \sqrt{2}$.
The runtime analysis is the same as before.

Since the runtime scales exponentially with the size of cubes $A_{j,p}$, one may wish
to choose a partition with smaller cubes even if this negatively impacts the  approximation quality.
As an extreme case, one can choose each cube $A_{j,p}$ as a single qubit. However ensuring the lightcone separation between
cubes in the same subset $A_j$ would require  $\approx (4h+1)^D$ subsets $A_j$ instead of $D+1$ subsets~\footnote{Since 
any qubit is lightcone separated from all but at most $(1+4h)^D$ other qubits, Vizing's theorem implies that  qubits can be partitioned into at most $1+(1+4h)^D$
lightcone separated subsets.}.
Accordingly, the approximation ratio  would become $\alpha=\Omega((4h+1)^D)$ instead of $\alpha=D+1$.

Likewise,
we  expect that the runtime can be improved at the cost of a worse approximation ratio $\alpha$
by computing the norm of commutators $K_{A_{j,p}}-I$ using a randomized version of the power method~\cite{kuczynski1992estimating}.
It is known that this method can approximate the operator norm 
of a matrix of size $2^m \times 2^m$ with a multiplicative error $1+\epsilon$ using
$O(m/\epsilon)$ matrix-vector multiplications~\cite{kuczynski1992estimating}.
In our case, $K_{A_{j,p}}$ is specified by a quantum circuit acting on $m=4(2hD)^D$ qubits
with $poly(m)$ gates,
see Section~\ref{sec:lightcones}. Thus one can implement matrix-vector multiplication for 
the matrix $K_{A_{j,p}}-I$ in time $poly(m) 2^m$.
Accordingly, the power method 
runs in time $poly(m) 2^m/\epsilon$, whereas the exact diagonalization 
of $K_{A_{j,p}}-I$ requires
time $\Omega(2^{3m})$.

\section{Numerical experiments}
\label{sec:numerics}

In this section, we implement the algorithm described in Section~\ref{sec:algo1}  to approximate the distance between identity and a constant-depth circuit $U$ of up to 100 qubits.
We consider $U = U_1 U_2^\dag$, where $U_1, U_2$ are two different unitaries that both approximate the time evolution $e^{-iH\tau}$ of $n$ qubits under the one-dimensional XY model:
\begin{align*} 
      H = \sum_{j = 1}^{n-1} \left(X_j X_{j+1} + Y_j Y_{j+1}\right).
\end{align*}
In the limit of small $\tau$, $U = U_1 U_2^\dag \approx I$ approximate a forward evolution followed by a backward evolution under the same Hamiltonian.
Explicitly, $U_1$ and $U_2$ are the first-order Trotter approximations with, respectively, an odd-even ordering and an X-Y ordering:
\begin{align*} 
   U_1 &= e^{-i\tau\sum_{\text{odd }j}\left(X_j X_{j+1} + Y_j Y_{j+1}\right) }e^{-i\tau\sum_{\text{even }j}\left(X_j X_{j+1} + Y_j Y_{j+1}\right)},\\
   U_2 &= e^{-i\tau\sum_{j}X_j X_{j+1}}e^{-i\tau\sum_{j} Y_j Y_{j+1} }.
\end{align*}
We note that $X_jX_{j+1}$ and $Y_jY_{j+1}$ are both antisymmetric under the unitary conjugation by the staggered Pauli string $X_1 Y_2 X_3 Y_4 \dots$.
Therefore, the eigenvalues of $U$ comes in complex conjugate pairs which results in a simple relationship between the diamond-norm and the operator-norm distances. Namely, a simple algebra shows that $\delta(U)=2\sin{(\varphi)}$, where
$\varphi\in [0,\pi/2)$ is defined by 
$\|U-I\|=|e^{i\varphi}-1|$.
In addition, using a well-known mapping from the XY model to free fermions~\cite{terhal2002classical}, we can compute this distance exactly, providing a benchmark for our algorithm.

\begin{figure}[b]
\centering
\includegraphics[width=0.45\textwidth]{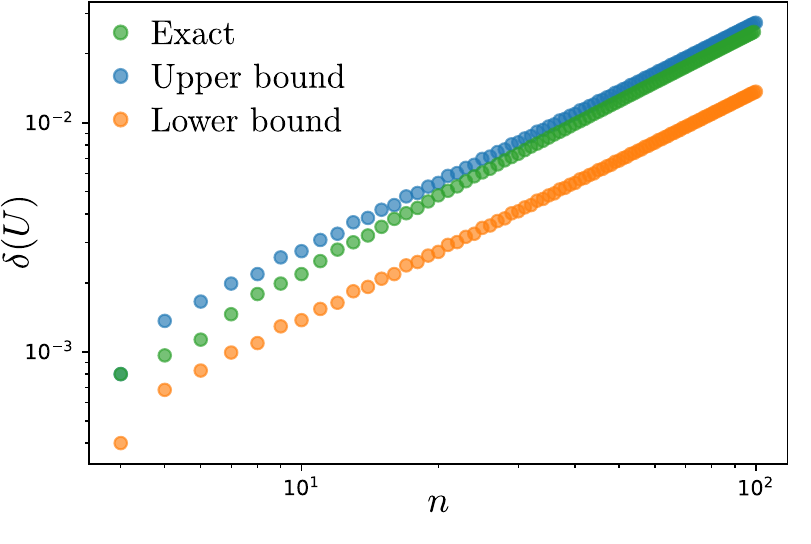}
\caption{A comparison between the exact diamond-norm distance 
$\delta(U)$ (green dots)
computed by a mapping to free fermions,
an upper bound $\gamma$
computed by Algorithm~1  (blue dots) and the  lower bound $\gamma/2$
(orange dots).
Both bounds closely capture the exact distance between $U$ and $I$, demonstrating the scalability of our algorithm.}
\label{fig:numerics}
\end{figure}

In Fig.~\ref{fig:numerics}, we compare the exact distance $\delta(U)$ against the bounds presented in Lemma~\ref{lemma:comm1} for up to 100 qubits at $\tau = 0.01$.
For the one-dimensional qubit array, the bounds simplify to $\delta(U) \leq \gamma \leq 2\delta(U)$, where
\begin{align} 
  \gamma =  \sum_{j=1}^2 \|  W_{A_j} (U \otimes I) W_{A_j}  (U^\dag \otimes I) - I\|.
\end{align}
Here, $A_1$ and $A_2$ are the qubit partitions illustrated in Fig.~\ref{fig:partition} with $L = 4$.
The lightcone separated construction of $A_j$ and the additivity lemma allow us to efficiently compute the commutator $\|  W_{A_j} (U \otimes I) W_{A_j}  (U^\dag \otimes I) - I\|$ for each $j$.
In particular, computing the bounds reduces to finding eigenvalues of operators that are each supported on at most 12 qubits.
Additionally, due to the translational invariance of the unitary $U$, only $O(1)$ such operators are unique, making the complexity of our algorithm independent of the system size.

Both bounds correctly capture the linear dependence of the Trotter error on the system size $n$, with
the upper bound $\gamma$ approaching the exact $\delta(U)$ in the limit of large $n$.
We note that $\norm{U-I}$ and, thus, $\delta(U)$ can also be estimated by finding the maximum eigenvalue of the Hamiltionian $H_U \equiv (U-I)^\dag (U-I)$.
Writing this Hamiltonian as a matrix product operator on a one-dimensional lattice, one can efficiently find a lower bound to its maximum eigenvalue using an algorithm based on the density matrix renormalization group (DMRG).
While DMRG does not have a performance guarantee, we find that it produces lower bounds to within $3\times 10^{-7}$ of the exact $\delta(U)$ in this example, providing a complementary approach to our algorithm in one dimension.
DRMG simulations were performed using  the matrix product representation library for Python
$\mathsf{mpnum}$
~\cite{mpnum} with 
MPS bond dimension $\chi=20$ and two DMRG sweeps in $\mathsf{mpnum.linalg.eig}$.

\section*{Acknowledgements}
SB thanks Steven Flammia and Kristan Temme for helpful discussions.
MCT thanks Kunal Sharma for helpful discussions.
This work was partially completed while NP was interning at IBM Quantum. NP is supported by AFOSR award FA9550-21-1-0040, NSF CAREER award CCF-2144219, and the Sloan Foundation.

\appendix

\section{Proof of Lemma~\ref{lemma:partition}}
\label{app:partition}

Let $A$ be an upper triangular $D\times D$ matrix with the unit diagonal.
In other words,
$A_{i,i}=1$ for all $i$ and $A_{i,j}=0$ for all $i>j$.
Define  a lattice $\calL_A\subseteq \RR^D$ formed by linear combinations
of columns of $A$ with integer coefficients. 
By definition, $p\in \calL_A$ iff $p=Ac$ for some integer vector $c\in \ZZ^D$.
For each lattice point $p\in \calL_A$ define an open cube $C(p)$ and a closed cube
$\overline{C}(p)$
such that $p$ is the cube's corner with the smallest coordinates, that is,
\[
C(p)=p+(0,1)^D \quad \mbox{and} \quad \overline{C}(p) = p + [0,1]^D.
\]
The following claim can be interpreted as saying that the cubes $C(p)$ form a partition of the Euclidean space $\RR^D$
if one ignores cube's boundaries.
\begin{claim}
Any point $x\in \RR^D$ is contained in at most one open cube $C(p)$.
Any point $x\in \RR^D$ is contained in at least one closed cube $\overline{C}(p)$.
\end{claim}
\begin{proof}
Define $\ell_\infty$ norm of a vector $x\in \RR^D$ as 
\[
\|x\|_\infty = \max_{i=1,\ldots,D} |x_i|.
\]
Suppose $x\in \RR^D$ is contained in cubes $C(p)$ and $C(q)$
for some lattice points $p,q\in \calL$. We have to show that $p=q$.
Clearly, cubes $C(p)$ and $C(q)$ overlap iff
\be
\label{pq_overlap}
\| p - q\|_\infty <1.
\ee
Thus we need to show that Eq.~(\ref{pq_overlap}) implies $p=q$.
Write
\be
\label{pq_overlap1}
r=p-q=Ac
\ee
for some $c\in \ZZ^D$. Using the upper triangular structure of $A$ and the fact that $A$
has unit diagonal one gets
\be
\label{pq_overlap2}
r_i = c_i + \sum_{j=i+1}^D A_{i,j} c_j.
\ee
If $i=D$ then clearly $r_i=c_i$ and thus $|r_i|<1$ is only possible if $c_i=0$.
If $i=D-1$ then $r_i=c_i + A_{i,D} c_D$. However, we have already showed that $c_D=0$.
Thus $r_i = c_i$ and $|r_i|<1$ is only possible if $c_i=0$.
Applying the same argument inductively proves that $c$ is the all-zeros vector, that is, Eq.~(\ref{pq_overlap}) implies $p=q$.

Suppose some vector $x\in \RR^D$ is not contained in any closed cube $\overline{C}(p)$.
Then $\|x-p\|_\infty>1$ for all lattice points $p\in \calL$.
Let us show that this assumption leads to a contradiction. Indeed,
set $i=D$.
Shift $x$ by an integer linear combination of the $i$-th column of $A$
to make $|x_i|\le 1$.
This is possible since $A_{i,i}=1$.
Next set $i=D-1$. Shift $x$ by an integer linear combination of the $i$-th column of $A$
to make $|x_i|\le 1$  and $|x_{i+1}|\le 1$.
This is possible since $A_{i,i}=1$ and $A_{i+1,i}=0$.
Applying the same argument inductively shows that shifting $x$ by lattice points
one can make $\|x\|_\infty\le 1$. Hence $x$ is contained in the cube $\overline{C}(0^D)$.
Equivalently, the original vector $x$ is contained in some
cube $\overline{C}(p)$.
\end{proof}

Following Ref.~\cite{woude2022geometry} we choose
\be
A_{i,j} = \left\{\ba{rcl}
1 &\mbox{if} & i=j,\\
\frac{D-j+1}{D} &\mbox{if} & i<j,\\
0 && \mbox{else}\\
\ea
\right.
\ee
for $1\le i,j\le D$.
For example,
\[
A=\left[\ba{cc}
1 & 1/2  \\
0 & 1 \\
\ea
\right]
\quad \mbox{and} \quad
A=\left[\ba{ccc}
1 & 2/3 & 1/3 \\
0 & 1 & 1/3 \\
0 & 0 & 1\\
\ea
\right]
\]
in the case $D=2$ and $D=3$ respectively. 
Below we summarize properties of the corresponding lattice $\calL_A$ established
in~\cite{woude2022geometry}. 
\begin{fact}[\bf Lemmas~7.15 and 7.19 of~\cite{woude2022geometry}]
The $\ell_\infty$-distance between closed cubes $\overline{C}(p)$ and $\overline{C}(q)$ is either $0$ 
(if these cubes overlap)
or at least $1/D$ (if these cubes do not overlap).
Here $p,q\in \calL_A$ are arbitrary lattice points.
\end{fact}
\begin{fact}[\bf Theorem~7.16 of~\cite{woude2022geometry}]
The cubes $\{\overline{C}(p)\}_{p\in \calL_A}$ can be colored with $D+1$ colors such that any cube $\overline{C}(p)$
overlaps only with 
cubes $\overline{C}(q)$ of a different color.
\end{fact}
As a consequence of Facts~1 and 2, the  $\ell_\infty$-distance between any pair of cubes  $\overline{C}(p)$
of the same color is at least $1/D$.
Rescaling each cube by the factor $L=2Dh$ and noting that $LA$ is an integer matrix
one obtains a partition of $\RR^D$ into a disjoin union of $D$-dimensional cubes $L\overline{C}(p)$
of linear size $L$ such that corners of each cube have integer  coordinates,
the cubes are colored with $D+1$ colors, and the  $\ell_\infty$-distance between any pair of cubes  
of the same color is at least $L/D$. 

Finally, embed a $D$-dimensional rectangular array into $\RR^D$
such that each cell of the array is a translation of the cube $(0,1)^D$ by an integer vector.
We can now define the desired set of cells $A_j$ as the union of all cells contained in the rescaled
cubes $L\overline{C}(p)$ of the $j$-th color.
This concludes the proof of Lemma~\ref{lemma:partition}.


\end{document}